\documentclass[letterpaper, 10 pt, conference]{ieeeconf}  

\usepackage{textcomp}
\usepackage{cite}
\usepackage{amsmath,amssymb,amsfonts}
\usepackage{graphicx}
\usepackage{url}
\usepackage{algpseudocode, algorithmicx,algorithm}
\usepackage{tikz}
\usetikzlibrary{shapes.geometric, positioning, calc}

\definecolor{planeTopCyber}{RGB}{218, 235, 244}
\definecolor{planeSideCyber}{RGB}{180, 210, 225}
\definecolor{planeTopPhys}{RGB}{225, 227, 233}
\definecolor{planeSidePhys}{RGB}{195, 200, 210}
\definecolor{planeTopAux}{RGB}{232, 226, 242}
\definecolor{planeSideAux}{RGB}{205, 198, 220}
\definecolor{nodeBlue}{RGB}{85, 135, 175}
\definecolor{nodeOrange}{RGB}{220, 130, 30}
\definecolor{darkgreen}{RGB}{0,120,0}

\usepackage{graphicx}
\usepackage{subcaption}
\IEEEoverridecommandlockouts                              

\title{\LARGE \bf
Decentralized Optimal Equilibrium Learning Over Dynamic Networks

}

\author{Seref Taha Kiremitci$^{1}$ and Muhammed O. Sayin$^{1}$%
\thanks{$^{1}$Seref Taha Kiremitci and Muhammed O. Sayin are with the Department of Electrical and Electronics Engineering, Bilkent University, Ankara, T\"urkiye. Emails: {\tt\small taha.kiremitci@bilkent.edu.tr}, {\tt\small sayin@ee.bilkent.edu.tr}}%
}

\newcounter{lemma}
\newenvironment{lemma}{\refstepcounter{lemma}
\noindent 
\textit{Lemma \thelemma.} \em \rmfamily}{}

\newcounter{claim}

\newcounter{proposition}
\newenvironment{proposition}{\refstepcounter{proposition}
\noindent 
\textit{Proposition \theproposition.} \em \rmfamily}{}

\newenvironment{proposition*}[1]{\refstepcounter{proposition}
\noindent 
\textit{Proposition \theproposition~(#1)} \em \rmfamily}{}

\newcounter{theorem}
\newenvironment{theorem}{\refstepcounter{theorem}
\noindent 
\textit{Theorem \thetheorem.} \em \rmfamily}{}

\newcounter{corollary}

\newcounter{definition}
\newenvironment{definition*}[1]{\refstepcounter{definition}
\noindent 
\textit{Definition \thedefinition~(#1)} \rmfamily}{}

\newcounter{problem}

\newcounter{assumption}
\newenvironment{assumption*}[1]{\refstepcounter{assumption}
\noindent 
\textit{Assumption \theassumption~(#1)} \rmfamily}{}
\newenvironment{assumption}{\refstepcounter{assumption}
\noindent 
\textit{Assumption \theassumption.} \em \rmfamily}{}

\newcounter{remark}
\newenvironment{remark}{\refstepcounter{remark}
\noindent 
\textit{Remark \theremark.} \em \rmfamily}{}

\newcounter{example}

\DeclareMathOperator*{\argmax}{arg\,max}
\newcommand{\be}{\begin{equation}}
\newcommand{\ee}{\end{equation}}
\newcommand{\nn}{\nonumber}
\newcommand{\E}{\mathrm{E}}

\newcommand{\SW}{\mathrm{SW}}
\newcommand{\eSW}{\overline{\mathrm{SW}}}

\newcommand{\explore}{\mathcal{E}}

\newcommand{\Aeq}{A_{\epsilon}^{\mathrm{eq}}}
\newcommand{\efu}{\overline{u}}
\newcommand{\efA}{\overline{A}}
\newcommand{\efa}{\overline{a}}

\usepackage[hidelinks]{hyperref}
\begin{document}

\maketitle
\thispagestyle{empty}
\pagestyle{empty}

\begin{abstract}

This paper studies decentralized learning of socially optimal equilibria in finite normal-form games over dynamic communication networks. Each agent observes only its own realized payoffs, does not know the game a priori, and can communicate only with time-varying neighbors using low-bandwidth messages. We propose networked decentralized optimal equilibrium learning dynamics in which agents generate randomized semantic content/discontent signals from local payoff comparisons and exchange time-stamped time-stacked tables rather than raw actions, payoff information or local estimates/parameters. The method combines table fusion with temporal majority reconstruction to mitigate dynamic communication while preserving fully decentralized operation. We establish finite-time logarithmic regret guarantees, with an in-phase exploration perturbation, for optimal equilibrium selection under utilitarian and proportional-fair social welfare objectives. Simulation results further show that the proposed approach can effectively select socially desirable equilibria over dynamic communication networks. 

\end{abstract}

\section{INTRODUCTION}

Decentralized multi-agent systems are increasingly central to modern engineered platforms, including networked control systems and distributed autonomous decision-making \cite{ref:Yuksel13,ref:Cao23b,ref:Basar98}. In such systems, a critical challenge is the balance between \emph{stability} and \emph{efficiency}. Equilibrium, providing robustness against unilateral deviations, need not be socially desirable. Welfare-maximization, being socially desirable, need not be stable under strategic behavior. Optimal equilibrium selection can balance stability and efficiency. Nevertheless, it is largely unexplored in the context of decentralized learning.

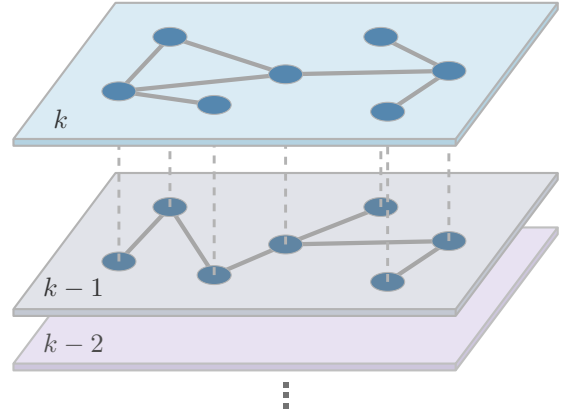
\begin{figure}[!t]
    \centering
\begin{tikzpicture}[scale=0.45,
    node default/.style={
        ellipse,
        fill=nodeBlue,
        draw=black!30,
        thin,
        minimum width=4.5mm,
        minimum height=2.5mm,
        inner sep=0pt
    },
    node default phys/.style={
        ellipse,
        fill=nodeBlue!75!gray,
        draw=black!30,
        thin,
        minimum width=4.5mm,
        minimum height=2.5mm,
        inner sep=0pt
    },
    edge normal/.style={draw=black!35, line width=1.6pt},
    edge failed/.style={draw=red!80, line width=1.8pt, dashed},
    edge inter/.style={draw=gray!60, line width=1.1pt, dashed}
]

\def\layersep{5}

\newcommand{\drawlayer}[3]{ 
    \begin{scope}[yshift=#1cm]
        \def\xbl{0} \def\ybl{0}
        \def\xbr{13} \def\ybr{0}
        \def\xtr{16} \def\ytr{4.0}
        \def\xtl{3} \def\ytl{4.0}
        \def\thick{0.2}

        \ifdim #1 pt > 0pt
            \colorlet{currSide}{planeSideCyber}
        \else
            \colorlet{currSide}{planeSidePhys}
        \fi

        \fill[currSide] (\xbl, \ybl-\thick) -- (\xbr, \ybr-\thick) -- (\xtr, \ytr-\thick) -- (\xtl, \ytr-\thick) -- cycle;
        \fill[currSide] (\xbr, \ybr) -- (\xbr, \ybr-\thick) -- (\xtr, \ytr-\thick) -- (\xtr, \ytr) -- cycle;
        \fill[currSide] (\xbl, \ybl) -- (\xbl, \ybl-\thick) -- (\xbr, \ybr-\thick) -- (\xbr, \ybr) -- cycle;

        \filldraw[fill=#2, draw=black!30, thick, line join=round]
            (\xbl, \ybl) -- (\xbr, \ybr) -- (\xtr, \ytr) -- (\xtl, \ytl) -- cycle;

        \draw[black!30, thick, line join=round]
            (\xbl, \ybl) -- (\xbl, \ybl-\thick) -- (\xbr, \ybr-\thick) -- (\xtr, \ytr-\thick) -- (\xtr, \ytr);
        \draw[black!30, thick, line join=round] (\xbr, \ybr) -- (\xbr, \ybr-\thick);

        \node[anchor=south west, font=\sffamily\Large, text=black!80]
            at (\xbl+0.6, \ybl+0.05) {#3};
    \end{scope}
}

\newcommand{\drawauxlayer}[2]{ 
    \begin{scope}[yshift=#1cm]
        \def\xbl{0} \def\ybl{0}
        \def\xbr{13} \def\ybr{0}
        \def\xtr{16} \def\ytr{4.0}
        \def\xtl{3} \def\ytl{4.0}
        \def\thick{0.2}

        \fill[planeSideAux] (\xbl, \ybl-\thick) -- (\xbr, \ybr-\thick) -- (\xtr, \ytr-\thick) -- (\xtl, \ytr-\thick) -- cycle;
        \fill[planeSideAux] (\xbr, \ybr) -- (\xbr, \ybr-\thick) -- (\xtr, \ytr-\thick) -- (\xtr, \ytr) -- cycle;
        \fill[planeSideAux] (\xbl, \ybl) -- (\xbl, \ybl-\thick) -- (\xbr, \ybr-\thick) -- (\xbr, \ybr) -- cycle;

        \filldraw[fill=planeTopAux, draw=black!25, thick, line join=round]
            (\xbl, \ybl) -- (\xbr, \ybr) -- (\xtr, \ytr) -- (\xtl, \ytl) -- cycle;

        \draw[black!25, thick, line join=round]
            (\xbl, \ybl) -- (\xbl, \ybl-\thick) -- (\xbr, \ybr-\thick) -- (\xtr, \ytr-\thick) -- (\xtr, \ytr);
        \draw[black!25, thick, line join=round] (\xbr, \ybr) -- (\xbr, \ybr-\thick);

        \node[anchor=south west, font=\sffamily\Large, text=black!70]
            at (\xbl+0.6, \ybl+0.0) {#2};
    \end{scope}
}

\def\nodes{
    \coordinate (N1) at (4.6, 3.0);
    \coordinate (N2) at (3.1, 1.4);
    \coordinate (N3) at (5.9, 1.0);
    \coordinate (N4) at (8.0, 1.9);
    \coordinate (N5) at (10.8, 3.0);
    \coordinate (N6) at (12.8, 2.0);
    \coordinate (N7) at (11.0, 0.8);
}

\drawauxlayer{-1.6}{\fontsize{10}{10}$k-2$}

\drawlayer{0}{planeTopPhys}{\fontsize{10}{10}$k-1$}
\begin{scope}[yshift=0cm]
    \nodes

    \draw[edge normal] (N1) -- (N3);
    \draw[edge normal] (N2) -- (N1);
    \draw[edge normal] (N4) -- (N5);
    \draw[edge normal] (N4) -- (N6);
    \draw[edge normal] (N6) -- (N7);
    \draw[edge normal] (N3) -- (N4);


    \foreach \i in {1,2,3,4,5,6,7} {
        \node[node default phys] at (N\i) {};
    }
\end{scope}

\begin{scope}
    \nodes
    \foreach \i in {1,...,7} {
        \draw[edge inter] (N\i) -- ([yshift=\layersep cm]N\i);
    }
\end{scope}

\drawlayer{\layersep}{planeTopCyber}{\fontsize{10}{10} $k$}
\begin{scope}[yshift=\layersep cm]
    \nodes

    \draw[edge normal] (N1) -- (N2);
    \draw[edge normal] (N2) -- (N3);
    \draw[edge normal] (N2) -- (N4);
    \draw[edge normal] (N5) -- (N6);
    \draw[edge normal] (N6) -- (N7);

    \draw[edge normal] (N4) -- (N1);
    \draw[edge normal] (N4) -- (N6);
    \foreach \i in {1,2,3,4,5,6,7} {
        \node[node default] at (N\i) {};
    }
\end{scope}
\node[font=\sffamily\huge, text=black!55] at (8.0,-2.3) {$\vdots$};
\end{tikzpicture}

\caption{An illustration of decentralized learning over dynamic communication networks. Layers depict time-varying topology at instances $k-2, k-1, k$.}    \label{fig:network}
\end{figure}

For example, in repeated play of normal-form games, decentralized payoff-based learning dynamics have been developed to promote Pareto-efficient solutions in interdependent games \cite{ref:Marden14}, and have been extended to the selection of Pareto-efficient Nash equilibria in distributed systems \cite{ref:Pradelski12}. Recently, \cite{ref:Yang25} has addressed decentralized learning of equilibrium maximizing the minimum transformed utility. All \cite{ref:Marden14,ref:Pradelski12,ref:Yang25} rely on stochastic stability guarantees as certain exploration parameters decay to zero. However, lower exploration levels lead to slower learning rates, limiting the practicality of such dynamics. Decentralized optimal equilibrium learning has also been studied in the context of stochastic games yet only for special classes of games, such as common-interest games and no-conflict games \cite{ref:Yongacoglu22,ref:Donmez25,ref:Christianos23}. All these works do not use any information exchange across agents while communication can play a crucial role in effective coordination \cite{ref:Nedic15,ref:Zhang18}. 


For example, \cite{ref:MardenWind13} studies welfare maximization through wind farm control, where turbines coordinate to maximize total power despite unknown interactions and have access to the aggregate output through communication. Their model-free scheme uses randomized experimentation with this shared feedback and provides asymptotic optimality guarantees for discretized controls under fixed wind conditions. Menon and Baras~\cite{ref:Menon13} address welfare optimization using local payoff observations and binary content/discontent messages, without communicating numerical payoff values or total welfare. They extend earlier payoff-based dynamics~\cite{ref:Marden14}, replacing the interdependence requirement with a connectivity condition on the combined interaction and communication graphs. They provide an asymptotic convergence guarantee under this condition and a suitable decreasing exploration schedule.

Rather recently, \cite{ref:Kiremitci26a} showed that single-bit communication can be effective to drive broad classes of general games toward socially optimal outcomes, while \cite{ref:Kiremitci26b} extended this idea to decentralized optimal equilibrium learning in stochastic games under general social welfare objectives. Both \cite{ref:Kiremitci26a} and \cite{ref:Kiremitci26b} use probabilistic \textit{content/discontent signals} rather than direct sharing of model parameters or estimates, and they provide explicit logarithmic finite-time guarantees. However, both of these schemes rely on fully connected and reliable communication, which is restrictive, e.g., in geographically distributed networks, and can make the learning mechanism sensitive to temporal link failures.

This paper addresses decentralized optimal equilibrium learning in repeated play of normal-form games over dynamic and unreliable communication links, as illustrated in Fig.~\ref{fig:network}. Agents observe only their own realized payoffs, do not know the game a priori, and communicate through finite-support semantic messages. Dynamic networks pose challenges that are not a direct extension of \cite{ref:Kiremitci26a,ref:Kiremitci26b}: content/discontent bits are phase-specific, so they cannot be averaged across time as in consensus, and neighbor-only failing links make network-wide bits delayed, missing, and statistically dependent across overlapping table windows. Our analysis quantifies finite-time learning performance in this dynamic-network setting through an expected-regret bound.

Relative to \cite{ref:Kiremitci26a,ref:Kiremitci26b}, which assume fully connected reliable broadcast of single-bit signals, our contributions are: (i) Networked-DOEL, which reconstructs network-wide content from time-stamped, time-stacked tables, conservative fusion, and temporal majority polling; (ii) a finite-time regret bound of the same form as in those works, but for dynamic neighbor-only graphs, by restoring independence of delayed access events rather than only rescaling success probabilities; (iii) experiments on a restricted Erd\H{o}s--R\'enyi backbone with dropouts, including the complete-graph communication setting of \cite{ref:Kiremitci26b}. We examine both utilitarian (sum of local payoffs) and proportional-fair (product of local payoffs) welfare objectives.

Our work is related more broadly to the large literature on networked coordination, but differs from it in an essential way. In distributed optimization and consensus over graphs, the network is typically used to aggregate estimates, gradients, or state variables toward a common point. This perspective includes distributed optimization or reinforcement learning over time-varying directed graphs \cite{ref:Nedic15,ref:Zhang18}, as well as communication-efficient schemes based on sign-based or compressed information exchange \cite{ref:Zhang19,ref:Sayin13,ref:Sayin14,ref:Cao23a}. Such works show that sparse and time-varying network structure can be exploited effectively, but their objectives are fundamentally cooperative and usually rely on structural properties such as convexity and smoothness, specifically for payoff-sum maximization. In contrast, our problem is game-theoretic. Here, the agents are strategically coupled and the target is reaching optimal equilibrium over the highly unstructured equilibrium set with respect to broad welfare objectives beyond payoff-sum maximization even under dynamic communication networks.

The rest of the paper is organized as follows. Section~\ref{sec:problem_formulation} introduces the problem formulation. Section~\ref{sec:alg_nfg} presents the proposed networked decentralized optimal equilibrium learning dynamics. Section~\ref{sec:main_result} and Section~\ref{sec:sim_game}, respectively, provide analytical and numerical results. Finally, Section~\ref{sec:conclusion} concludes the paper.

\section{PROBLEM FORMULATION}\label{sec:problem_formulation}

Consider an $n$-agent \emph{normal-form} game characterized by the tuple
$\mathcal{G} := \langle A^i,u^i\rangle_{i\in\mathcal N}$,
where $\mathcal N:=\{1,\ldots,n\}$ is the index set of agents,
$A^i$ is the finite action set of agent $i$,
$A:=\prod_{j\in\mathcal N}A^j$ is the joint action space,
and $u^i:A\to\mathbb R$ is agent $i$'s payoff function.

Let $a=(a^i)_{i\in\mathcal N}\in A$ be the action profile of all agents and $a^{-i}:=(a^j)_{j\neq i}$ denote actions of agents other than agent $i$. Then, an action profile $a\in A$ is a \emph{pure-strategy $\epsilon$-Nash equilibrium} if, for each agent $i$,
\begin{equation}\label{eq:best_nf}
u^i(a^i,a^{-i})
\ge u^i(\tilde a^i,a^{-i})-\epsilon^i,
\quad \forall \tilde a^i\in A^i,
\end{equation}
for some $\epsilon^i\ge 0$.
Given $\epsilon=(\epsilon^i)_{i=1}^n$, we define $\epsilon$-equilibrium set as
\begin{equation}\label{eq:eqset_nf}
A_{\epsilon}^{\mathrm{eq}}
:= \Big\{a\in A:
u^i(a) \ge \max_{\tilde a^i\in A^i}u^i(\tilde a^i,a^{-i})-\epsilon^i,\ \forall i\Big\}.
\end{equation}

In general, pure-strategy equilibrium may not exist. However, there always exist sufficiently large tolerance levels $\epsilon$ such that a pure-strategy $\epsilon$-equilibrium can exist. Furthermore, different tolerance levels can capture heterogenous sensitivities to suboptimal responses, ranging from full non-cooperativeness as $\epsilon^i\rightarrow 0$ to full cooperativeness as $\epsilon^i\rightarrow\infty$.

Given the payoffs $(u^i)_{i=1}^n$, the social welfare of all agents is given by
\begin{equation}\label{eq:SW}
\SW(a):=\sum_{i=1}^n w^i\cdot g^i\!\big(u^i(a)\big),
\end{equation}
for positive weights $w^i>0$ and monotonically increasing transformations $g^i(\cdot)$.

For example, for $w^i=1$ and $g^i(x) = x$, the social welfare corresponds to the sum of payoffs $\sum_{i} u^i(a)$, studied in multi-agent systems for Pareto optimal solutions, e.g., see \cite{ref:Marden14,ref:Pradelski12,ref:Nedic15,ref:Zhang18}. On the other hand, for $w^i=1$ and $g^i(x) = \log (x)$, the social welfare corresponds to the product of payoffs $\log\prod_{i} u^i(a)$, yielding fair solutions across agents \cite{ref:Kelly98}.  

We say that equilibrium $a_\ast\in A_{\epsilon}^{\mathrm{eq}}$ is \emph{optimal} over (non-empty) $\epsilon$-equilibrium set $\Aeq\neq \varnothing$ provided that
\begin{equation}\label{eq:optimal_nf}
a_\ast \in \argmax_{a\in \Aeq}\left\{\SW(a)\right\}.
\end{equation}

We consider that agents repeatedly play the game $\mathcal{G}$ to learn optimal equilibria through decentralized learning dynamics. They do not know the model and do not observe the actions of the others. They only observe their own payoff with perfect recall. We let agents communicate with certain other agents (e.g., the ones in proximity) over a communication network through finite support messages while the communication links can change or fail dynamically, as illustrated in Fig. \ref{fig:network}. 

In this paper, our goal is to \textit{design resilient decentralized optimal equilibrium learning algorithms for such dynamic communication networks}. To quantify the effectiveness of the proposed solution, we focus on regret-based analysis. Particularly, given a time budget $T$, we define the cumulative regret of the agents' joint play $a_t$ at time $t$ relative to the optimal welfare in~\eqref{eq:optimal_nf} as
\begin{equation}\label{eq:regret}
R_T :=
T\!\max_{a\in A_{\epsilon}^{\mathrm{eq}}}
\left\{\SW(a)\right\}
 - \sum_{t=1}^{T}\SW(a_t).
\end{equation}

\section{DECENTRALIZED OPTIMAL EQUILIBRIUM LEARNING OVER NETWORKS}
\label{sec:alg_nfg}

In this section, we first describe Decentralized Optimal Equilibrium Learning (DOEL) framework, in which each agent can reliably communicate with all other agents \cite{ref:Kiremitci26b}. Although agents only exchange single-bit messages, as described in Subsection \ref{sec:DOEL}, full connectivity poses a challenge for practical applications. For example, there may not be a direct (one-hop) communication link in-between all agents. To mitigate these issues, we introduce Networked-DOEL dynamics, in which agents can temporally communicate with only a subset of agents over a dynamic communication network.

\subsection{DOEL Dynamics}\label{sec:DOEL}

The DOEL dynamics originally address decentralized equilibrium selection for stochastic games through explore-and-commit and online learning schemes. However, due to the limited space, here, we only focus on the explore-and-commit scheme for the special case of repeated play of normal-form games, in which the underlying environment has a single state and agents have single-stage objectives. 

During the exploitation stage, agents aim to coordinate in the best $\epsilon$-equilibrium. To this end, they need to identify (i) whether the explored action profile is $\epsilon$-equilibrium and (ii) whether it leads to the highest social welfare compared to other equilibria. 

To address (i), agents can divide the exploration stage into $K$ phases $k=1,\ldots,K$ with fixed lengths of $\kappa$ instances.  At each phase $k$, each agent $i$ randomly selects an action $a_{(k)}$. Then, within the phase $t=(k-1)\kappa + 1,\ldots, k\kappa$, the agent plays according to the distribution $\explore^i(a_{(k)}^i):A^i\rightarrow[0,1]$ defined by
\be
\explore^i(a_{(k)}^i)(a^i) = \left\{\begin{array}{ll}
1-\frac{|A^i|-1}{|A^i|} \cdot \varepsilon^i& \mbox{if } a^i = a_{(k)}^i\\
\frac{1}{|A^i|}\cdot\varepsilon^i & \mbox{if } a^i\neq a_{(k)}^i
\end{array}\right.,
\ee
where $\varepsilon^i\in (0,1)$ is the \textit{in-phase exploration} probability. Correspondingly, $a_t^i\sim \explore^i(a_{(k)}^i)$ independently for each $t=(k-1)\kappa + 1,\ldots, k\kappa$. 

Each agent can estimate the value of local actions by taking average of the payoffs received whenever the associated action taken. For example, we have
\be\label{eq:ave}
    \hat{u}_{(k)}^i(a^i) = \frac{\sum_{t=(k-1)\kappa+1}^{k\kappa} u_t^i \cdot \mathbb{I}_{\{a^i = a_t^i\}}}{\sum_{t=(k-1)\kappa+1}^{k\kappa} \mathbb{I}_{\{a^i = a_t^i\}}}
    \ee
for all $a^i$. Then, explored $a_{(k)}^i$ is the $\epsilon$-best response with respect to the estimate $\hat{u}_{(k)}^i$ if 
\be\label{eq:bestest}
\hat{u}_{(k)}^i(a_{(k)}^i) \geq \hat{u}_{(k)}^i(a^i) - \epsilon^i\quad\forall a^i\in A^i.
\ee
Furthermore, the explored action profile $a_{(k)}=(a_{(k)}^i)_{i\in \mathcal{N}}$ is $\epsilon$-equilibrium if \eqref{eq:bestest} holds for all $i\in \mathcal{N}$.

Recall that agents only observe local actions and payoffs. Even though they can check whether the action explored tolerable, they are not aware of whether everyone's explored actions are tolerable. To mitigate this issue, they can broadcast a single-bit tolerable $'1'$ or intolerable $'0'$ signal with each other at the end of the phase. If everyone signals $1$, then the agents can conclude that the explored action profile is $\epsilon$-equilibrium based on the local value estimates.


\begin{remark}
If agents are self-less, i.e., $\epsilon^i\rightarrow\infty$, then the problem reduces to decentralized welfare optimization and they do not need to divide the exploration stage into phases.  
\end{remark}


To address (ii), agents need to consider the social welfare \eqref{eq:SW}. However, local payoffs are not necessarily aligned with the social welfare.  In the DOEL dynamics, agents broadcast a \textit{randomized} single-bit \textit{content} $'1'$ or \textit{discontent} $'0'$ signal with each other if the local action is tolerable.\footnote{In the payoff-based dynamics of \cite{ref:Marden14,ref:Pradelski12}, agents locally keep and update a \textit{content/discontent} state without sharing it with the other agents.} More explicitly, if
$\hat{u}_{(k)}^i(a_{(k)}^i) \geq \max_{a^i} \hat{u}_{(k)}^i(a^i) - \epsilon^i$, then agent $i$ broadcasts
\be
s_{(k)}^i = \left\{\begin{array}{ll}
1 & \mbox{with probability } \xi^{w^i(C^i-g^i(\hat{u}_k^i(a_{(k)}^i)))}\\
0 & \mbox{otherwise}
\end{array}\right.
\ee
independently with some $\xi\in(0,1)$ and $C^i \geq g^i(u^i(a))$ for all $a\in A$, ensuring that $C^i-g^i(\hat{u}_k^i(a_{(k)}^i))\geq 0$ and $\xi^{w^i(C^i-g^i(\hat{u}_k^i(a_{(k)}^i)))}\in(0,1]$ is a valid probability. Due to the independent randomization, all agents would broadcast content signals for $\epsilon$-equilibrium $a_{(k)}$ based on the local value estimates with the conditional probability
\be\label{eq:prob}
\Pr(s_{(k)}^i=1\;\forall i \mid a_{(k)}) = \xi^{C-\SW(a_{(k)})},
\ee
where $C= \sum_i w^i C^i$ and $C\geq \SW(a_{(k)})$ by its definition. 

Observe that $a_{(k)}$ maximizing the probability \eqref{eq:prob} maximizes $\SW(\cdot)$ since $\xi\in(0,1)$. Indeed, the distribution over action profiles conditioned on content-endorsement by all agents corresponds to the entropy-regularized welfare maximization over the equilibrium set \cite[Remark 2]{ref:Kiremitci26b}. 

Another key challenge is that agents do not observe the explored action profile. Relying only on local observations and communication, the DOEL dynamics let agents commit to the local action that is the most frequently content-endorsed by all agents. Under mild assumptions, \cite{ref:Kiremitci26b} showed that the DOEL dynamics can attain at most logarithmic regret up to an in-phase exploration offset. 

\begin{figure*}[t!]
    \centering
    \includegraphics[width=\textwidth]{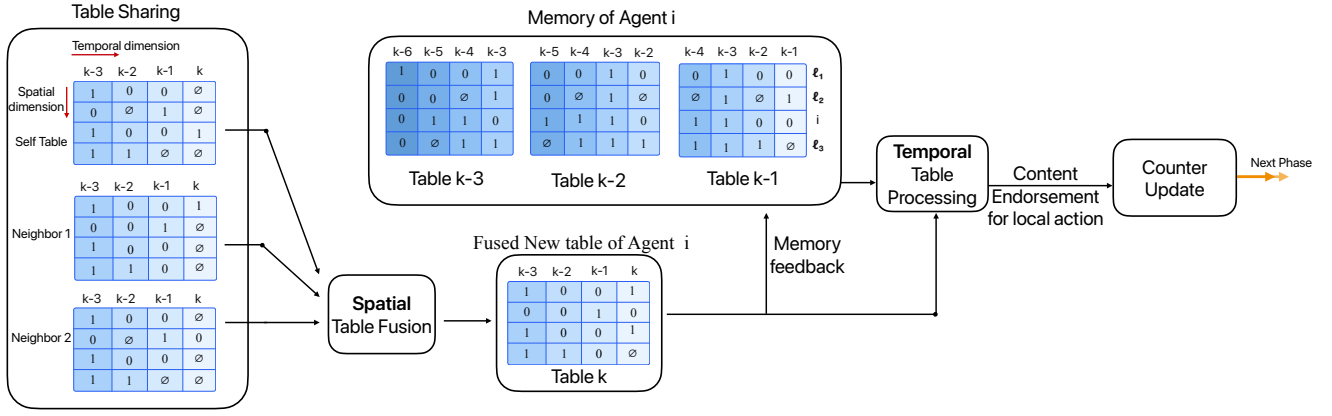}
\caption{A figurative illustration of spatial table fusion and temporal table processing procedures for agent \(i=1\) at phase \(k\) over the four-agent communication network depicted at the top-right corner. The agent construct the table \(o_{(k)}^i\) based on the current content/discontent signal $s_{(k)}^i$ and the previous beliefs $\{\hat{s}_{(k),(k-l)}^{i,j}\}$ about the content/discontent signals of all agents $j\in \mathcal{N}$ and at previous phases $l=k-m+1,\ldots,k-1$. The agent receives the neighbor tables \(\{o_{(k)}^j\}_{j\in\mathcal{N}_{(k)}^i}\) and fuse them with the local table \(o_{(k)}^i\) to compute  \(\overline{o}_{(k)}^i\). Then, the agent uses \(\overline{o}_{(k)}^i\) and  \(\{\overline{o}_{(k-l)}^i\}_{l=1}^{n-1}\) to determine whether all agents were content at phase $k-n+1$ through majority polling.}
\label{fig:algoritma}
\end{figure*}

\subsection{Networked-DOEL Dynamics}\label{sec:RN-DOEL}

In the Networked-DOEL dynamics, agents are not able to reliably broadcast signals to all other agents. The communication network might dynamically change across phases, e.g., see Fig. \ref{fig:network}. To resolve this issue, the Networked-DOEL dynamics let agents share time-stamped and time-stacked messages, as illustrated in Fig. \ref{fig:algoritma}.

\textit{Time-Stamped Signals:} Content/discontent signals are specific to the current actions explored. However, the actions independently got explored  according to the uniform distribution across phases. Therefore, agents are not able to diffuse the content signals across the underlying communication network by aggregating all information shared as in consensus networks. Instead, they can share the content/discontent signals with a timestamp to identify the phases associated with the signals. They not only share their own content/discontent signals but also share the most recent content/discontent signals of all others based on the received signals from the neighboring agents, as commonly used for dynamic environments, e.g., see \cite{ref:Mitra22}. 

\textit{Time-Stacked Signals:} Here all signals matter beyond the most fresh ones. If the signals of agents do not reach to each other, then these agents cannot track the network-wide signals and this can degrade the performance. To mitigate this issue, agents can introduce redundancy by sharing not only the most recent content/discontent signals but also older signals.

More explicitly, let $\hat{s}_{(k),(l)}^{i,j}\in \{0,1,\varnothing\}$  denote agent $i$'s belief at phase $k$ about agent $j$'s content/discontent signal at phase $l$. Here, $\varnothing$ refers to missing or inconsistent information. Then, at phase $k$, agent $i$ shares an $n \times m$ dimensional matrix
\be\label{eq:sharedtable}
o_{(k)}^i = \begin{bmatrix} 
\hat{s}_{(k),(k-m+1)}^{i,1} & \cdots &\hat{s}_{(k),(k)}^{i,1} \\
\vdots & \ddots & \vdots\\
\hat{s}_{(k),(k-m+1)}^{i,n} & \cdots &\hat{s}_{(k),(k)}^{i,n}
\end{bmatrix} \in \{0,1,\varnothing\}^{n\times m}.
\ee
This only creates an $n\times m$-bit memory burden per agent, which can be further reduced, thereby yielding a scalable memory footprint.
If the underlying communication networks consistently have depth more than $m$, then some agents may not access to the signals of all others. On the other hand, a connected network can have the depth of at most $n-1$. Correspondingly, $m\geq n-1$ can ensure that all agents can access to the signals of all others as long as the networks remain connected.

\textit{Spatial Table Fusion:} Agents fuse the tables received by the neighboring agents to construct 
\begin{subequations}\label{eq:fusedtable}
\begin{flalign}
&\overline{o}_{(k)}^i = [\overline{s}_{(k),(k-l)}^{i}]_{l=k-m+1}^k\\
&\overline{s}_{(k),(l)}^{i} = [\overline{s}_{(k),(l)}^{i,j}]_{j\in\mathcal{N}},
\end{flalign}
\end{subequations}
as depicted in Fig. \ref{fig:algoritma}. Formally, we can represent the communication network at phase $k$ by an undirected graph $G_k = (\mathcal{N},E_{(k)})$, where $\mathcal{N}$ is the set of vertices/nodes corresponding to agents and $E_{(k)}\subset \mathcal{N}\times\mathcal{N}$ is the set of edges depicting the communication links. For example, the pair $(i,j),(j,i)\in E_{(k)}$ if agents $i$ and $j$ can exchange messages at phase $k$. We denote the neighborhood of agent $i$ at phase $k$ by $\mathcal{N}_{(k)}^i := \{j\in \mathcal{N}: (i,j)\in E_{(k)}\}$. Then, spatial table fusion can accept signals consistent across tables $\{o_{(k)}^r\}_{r\in \mathcal{N}_{(k)}^i}$ and
\be\label{eq:interbelief}
\overline{s}_{(k),(l)}^{i,j} = \left\{\begin{array}{ll}
0 & \mbox{if } \hat{s}_{(k),(l)}^{r,j} = 0 \;\forall r\in \mathcal{N}_{(k)}^i : \hat{s}_{(k),(l)}^{r,j}\neq \varnothing \\
1 & \mbox{if } \hat{s}_{(k),(l)}^{r,j} = 1 \;\forall r\in \mathcal{N}_{(k)}^i : \hat{s}_{(k),(l)}^{r,j}\neq \varnothing,\\
\varnothing & \mbox{otherwise}.
\end{array}\right.
\ee

Furthermore, at the next phase, the beliefs get updated as
\be\label{eq:beliefupdate}
\hat{s}_{(k+1),(k-l)}^{i,j} = \left\{\begin{array}{ll}
\overline{s}_{(k),(k-l+1)}^{i,j} & \mbox{if } l>0\\
s_{(k+1)}^i &\mbox{if } l = 0 \mbox{ and } j=i\\
\varnothing &\mbox{otherwise}.
\end{array}\right.
\ee
to construct the table $o_{(k+1)}^i$ to be shared with the neighboring agents.

\textit{Temporal Processing of Tables:} Tables $\{o_{(k)}^r\}_{r\in\mathcal{N}}$ received at phase $k$ include beliefs $\hat{s}_{(k),(k-l)}^{r}$ about older phases $k-l$. By considering older versions of these beliefs, agents can mitigate the impact of missing links. To this end, agents can store tables fused in memory and memory size is proportional to the table depth $m$. 

For example, consider the example depicted in Fig. \ref{fig:algoritma} with four agents. At phase $k$, table $\overline{o}_{(k)}^i$ includes the earliest beliefs $\hat{s}_{(k),(k-3)}^i$ while any tables $\overline{o}_{(k-l)}^i$ for $l>3$ only consist of beliefs earlier than that, making these tables redundant to store. Hence, agents can have finite memory in which they store tables $\{\overline{o}_{(k-l)}^i\}_{l=1}^{m-1}$.

Given the tables $\{\overline{o}_{(k-l)}^i\}_{l=1}^{m-1}$ stored in the memory and the new table $\overline{o}_{(k)}^i$, agent $i$ can use \textit{majority polling} over $\{\overline{s}_{(k-l),(k-m+1)}^i\}_{l=0}^{m-1}$ to determine if all were content at phase $k-m+1$. 
Define $b_{(l),(k)}^i=\varnothing$ if $\{j\in \mathcal{N} \;:\; \overline{s}_{(l),(k)}^{i,j}\neq \varnothing\}$ is an empty set, and otherwise
\be
b_{(l),(k)}^i := \prod_{j\in \mathcal{N} \;:\; \overline{s}_{(l),(k)}^{i,j}\neq \varnothing} \overline{s}_{(l),(k)}^{i,j},
\ee
which is one if all information signals $\overline{s}_{(l),(k)}^{i,j}$ are ones. Then, the agent detects network-wide content endorsement for the local explored action $a_{(k-m+1)}^i$ if
\be\label{eq:polling}
\sum_{l=0}^{m-1} b_{(k-l),(k-m+1)}^i \geq \frac{1}{2}\sum_{l=0}^{m-1} \mathbb{I}_{\{b_{(k-l),(k-m+1)}^i\neq \varnothing\}}.
\ee
In other words, the local action gets network-wise content-endorsement if network-wide content endorsement get detected at least half of the network-wide signals available.

\begin{algorithm}[t!]
\caption{Networked-DOEL Dynamics}
\label{alg:main}
\begin{algorithmic}
\small
\Require{$\xi\in(0,1)$, $\epsilon^i\ge 0$, $\varepsilon^i\in(0,1)$, $C^i \geq \max_u g^i(u)$}
\State \hspace{0.5em}\textbf{Initialize:} counter $c^i(a^i)=0$ for all $a^i\in A^i$     
\Statex
    \tikz[remember picture,overlay] {
        \node[rotate=90,anchor=south,yshift=5pt] at (0.2,-3.5) {\scriptsize \textbf{Exploration}};
        \draw[thick] (0,-7.2) -- (0,0.1);
    }
    \vspace{-0.3cm}
\State \hspace{0.5em}\textbf{for} each phase $k=1,\ldots,K$ \textbf{do}
    \State \hspace{2em}explore local action $a_{(k)}^i \sim \mathrm{Uniform}(A^i)$ independently
    \State \hspace{2em}\textbf{for} each stage $t=(k-1)\kappa+1,\ldots,k\kappa$ \textbf{do}
        \State \hspace{4em}play $\tilde a_t^i\sim \explore^i(a_{(k)}^i)$ \Comment{Simultaneous play}
        \State \hspace{4em}receive payoff $u_t^i$
    \State \hspace{2em}\textbf{end for}
    \State \hspace{2em}compute $\hat{u}_{(k)}^i(a^i)$ for all $a^i$ as in \eqref{eq:ave}
    \State \hspace{2em}compute $\hat{u}_{*,(k)}^i = \max_{a^i} \hat{u}_{(k)}^i(a^i)$ 
    \State \hspace{2em}generate content/discontent signal:
    \[\hspace{2em}
    s_{(k)}^i=
    \begin{cases}
    \mathbb I_{\{\hat{u}_{(k)}^i(a_{(k)}^i) \ge \hat u_{*,k}^i-\epsilon^i\}} & \text{w.p. } \xi^{\,w^i(C^i-g^i(\hat{u}_k^i(a_{(k)}^i)))}\\
    0 & \text{o.w.}
    \end{cases}
    \]
    \State \hspace{2em}construct the table $o_{(k)}^i$ according to \eqref{eq:sharedtable} and \eqref{eq:beliefupdate}
\State \hspace{2em}share $o_{(k)}^i$ and receive $\{o_{(k)}^j\}_{j\in\mathcal{N}_{(k)}^i}$
\State \hspace{2em}construct $\overline{o}_{(k)}^i$ by fusing $o_{(k)}^i$ and $\{o_{(k)}^j\}_{j\in\mathcal{N}_{(k)}^i}$ as in \eqref{eq:interbelief}
\State \hspace{2em}\textbf{if} $k \geq m$ and \eqref{eq:polling} holds \textbf{then}
\State \hspace{4em}increment the counter $c^i(a_{(k-m+1)})$ by one 
\State \hspace{2em}\textbf{end if}
\State \hspace{0.5em}\textbf{end for}
    \vspace{-0.3cm}
    \Statex
    \tikz[remember picture,overlay] {
        \node[rotate=90,anchor=south,yshift=5pt] at (0.2,-1.0) {\scriptsize \textbf{Exploitation}};
        \draw[thick] (0,-1.8) -- (0,-0.1);
    }
\State \hspace{0.5em}identify $\hat a^i_{(K)} \in \arg\max_{a^i\in A^i}\{c_{(K)}^i(a^i)\}$

\Statex \hspace{0.5em}\textbf{for} each stage $t > \kappa K$ \textbf{do}
    \State \hspace{2em}play $a_t^i = \hat a^i_{(K)} $ \Comment{Simultaneous play}
    \State \hspace{2em}receive payoff $u_t^i$
\Statex \hspace{0.5em}\textbf{end for}
\end{algorithmic}
\end{algorithm}

Conservative table fusion accepting only signals consistent across neighboring agents and majority polling over network-wide available signals reduce the sensitivity of the dynamics to erroneous signal transmissions beyond link failures. 

Algorithm \ref{alg:main} tabulates the Networked-DOEL dynamics for the typical agent $i$.

\section{ANALYTICAL RESULTS}
\label{sec:main_result}
In this section, we characterize the expected regret \eqref{eq:regret} for the Networked-DOEL dynamics over dynamic communication networks. To this end, we consider that the underlying communication network $G_k$ at phase $k$ gets generated \textit{exogenously} according to the Erdős--Rényi $G(n,p)$ model. There can be a direct communication link between any pair of agents with probability $p\in(0,1]$ independent of other pairs, previous networks, and the Networked-DOEL dynamics. 

Let $q_m^i\in[0,1]$ denote the probability that agent $i$ can access to the signals of all others within $m$ consecutive phases. Across $m$ consecutive phases, there may not be any one-length, i.e., direct, paths between agent $i$ and some other agent $j$ with probability $(1-p)^m$. This yields that
\be\label{eq:q_bound}
q_m^i \geq (1-(1-p)^m)^{n-1} =: q_{m}
\ee
and $q_{m}>0$ since $p\neq 0$. Hence, agents can access to network-wide signals $\{s_{(k)}^i\}_{i\in \mathcal N}$ associated with any phase $k$ with at most $m$-phase delay with some \textit{positive} probability. Indeed, information can also reach in multiple hops. However, \eqref{eq:q_bound} provides a tractable lower bound, growing with the table width $m$ while decaying with the number of agents and edge probability $p$. 

The Networked-DOEL dynamics involve in-phase exploration so that agents can identify whether the explored action is tolerable or not. However, since the other agents also conduct in-phase exploration, agents effectively identify whether the explored action is tolerable against exploration-perturbed actions of others. Therefore, we introduce the \textit{effective payoff} $\efu^i(a^i,a^{-i}) := u^i(a^i,\explore^{-i}(a^{-i}))$ for all $i$ and $(a^i,a^{-i})$, where $\explore^{-i}(a^{-i})=(\explore^j(a^j))_{j\neq i}$. Given the effective payoffs, we define the effective welfare and the effective $\epsilon$-equilibrium set of the corresponding \textit{effective game} $\overline{\mathcal{G}}=\langle A^i,\efu^i\rangle_{i\in \mathcal{N}}$, resp., by
\begin{flalign}
&\eSW(a)=\sum_{i=1}^n w^i\cdot g^i(\efu^i(a)),\\
&\efA_{\epsilon}^{\mathrm{eq}}
= \Big\{a\in A: \efu^i(a) \ge \max_{\tilde a^i\in A^i}\efu^i(\tilde a^i,a^{-i})-\epsilon^i,\ \forall i\Big\}.
\end{flalign}

We make the following assumptions:

\begin{assumption}\label{assm:equilibrium}
\begin{itemize}
\item[(i)] The effective $\epsilon$-equilibrium set $\efA_{\epsilon}^{\mathrm{eq}}$ is non-empty.
\item[(ii)] There exists unique effective equilibrium $\efa^*\in \efA_{\epsilon}^{\mathrm{eq}}$ maximizing the effective welfare $\eSW$, i.e.,
\be
\eSW(\efa^*) > \eSW(\efa)\quad\forall \efa\in \efA_{\epsilon}^{\mathrm{eq}}\mbox{ and } \efa \neq \efa^*.
\ee
\item[(iii)] Let $\Xi_o>0$ denote the gap between the best and second best(s) as
\be
\Xi_o := \eSW(\efa^*) - \max_{\efa\in \efA_{\epsilon}^{\mathrm{eq}}\setminus\efa^*} \{\eSW(\efa)\}.
\ee 
Then, for some $\delta > 0$, we have
\be
0 < \xi < (|\efA_{\epsilon}^{\mathrm{eq}}|+\delta)^{-1/\Xi_o}. 
\ee
\end{itemize}
\end{assumption}

This assumption ensures that agents can learn the unique optimal equilibrium in a decentralized way through randomized content/discontent signals. 

\begin{assumption}\label{assm:estimation}
The phase length $\kappa$ is sufficiently long such that in-phase estimation errors are negligible, i.e., $\hat{u}_{(k)}^i(a^i_{(k)}) = \efu^i(a^i_{(k)},\efa_{(k)}^{-i})$ and $\hat{u}_{*,(k)}^i = \max_{a^i\in A^i}\efu^i(a^i,\efa_{(k)}^{-i})$.
\end{assumption}

The estimates \(\hat{u}_{(k)}^i\) in \eqref{eq:ave} use only agent \(i\)'s own in-phase payoffs, so they do not depend on \(G_k\). Ordinary concentration therefore applies directly, as analyzed in \cite{ref:Kiremitci26b}: Assumption 2 holds whenever \(\kappa\) is large enough for empirical payoffs to concentrate. We use this exact-equality form to isolate communication error in Theorem 1. In real systems, phases must be long enough relative to payoff variability; otherwise content bits can be wrong and the bound does not apply.

The following theorem shows that the Networked-DOEL dynamics have logarithmic regret with an in-phase exploration perturbation.

{\medskip}

\begin{theorem}
Consider the repeated play of a normal-form game $\mathcal{G}=\langle A^i,u^i\rangle_{i\in \mathcal{N}}$. Suppose that each agent follows Algorithm \ref{alg:main} and Assumptions \ref{assm:equilibrium} and \ref{assm:estimation} hold. Furthermore, at each phase, let the communication network be generated according to the Erdős--Rényi $G(n,p)$ model. If the exploration length be $K\in O(\log T / q_m^2)$, then the expected regret is bounded by $\E[R_T] \in O(\varepsilon T + \log T)$.
\end{theorem}

{\medskip}

The Networked-DOEL dynamics use in-phase $\varepsilon$-exploration to identify whether the explored actions are tolerable. However, this causes an offset on the regret. Furthermore, the exploration length $K\in O(\log T / q_m^2)$ is reversely proportional to the square of the probability lower bound $q_{m}$. As described in \eqref{eq:q_bound}, $q_{m}\in(0,1)$ increases with larger table size $m$, larger connectivity probability $p$ and smaller number of agents.

{\medskip}

\begin{proof}
We can write the expected regret as
\be\nn
\E[R_T] = T\SW^* - \bigg(\kappa\sum_{k=1}^K \SW(\efa_{(k)})+ (T-\kappa K) \SW(\hat{a}_{(K)})\bigg),
\ee
where $\efa_{(k)} = (\explore^i(a_{(k)}^i))_{i\in\mathcal{N}}$ is the exploration-perturbed action profile. Define $\Xi:= \SW^*-\min_a \SW(a)$ and $\Xi_\varepsilon := \SW^* - \SW(\efa^*)$. Then, we have
\be
\E[R_T] \leq \kappa K \Xi + T \Xi_\varepsilon + T\Xi\Pr(\hat{a}_{(K)}\neq \efa^*).
\ee 
Since the repeated play of normal-form games is a special case of stochastic games, \cite[Lemma 1]{ref:Kiremitci26b} yields that $0<\Xi_\varepsilon \leq \bar{C} \varepsilon$ for certain $\bar{C}>0$.

Next, we bound the probability $\Pr(\hat{a}_{(K)}\neq \efa^*)$ from above. The challenge is that agents may not access to the network-wide signals. Furthermore, whether agent $i$ can access to network-wide signal at phase $k$ is not independent of having access to network-wide signals at phases $\{k-m+1,\ldots,k+m-1\}$ since the shared tables consist of signals within a window of $m$ phases. To address this issue, we fix agent $i$ and consider $m$ separate empirical averages:
\be\label{eq:hattheta}
\hat{\theta}_{(\ell),l}^{(i)}(a):= \frac{1}{\ell}\sum_{k=1}^\ell \mathbb{I}_{\{a_{mk + l} = a\}} \cdot s_{(mk+l)} \cdot r_{(mk+l)}^{(i)}
\ee
for $l=0,1,\ldots,m-1$, where $s_{(mk+l)} = \prod_{j} s_{(mk+l)}^j \in \{0,1\}$ is $1$ if there is network-wide content endorsement and $0$ otherwise, and $r_{(mk+l)}^{(i)}\sim \mathrm{Bernoulli}(q_m^i)$ is the indicator of whether agent $i$ can get access to the network-wide signal possibly with some delay over the randomly selected communication networks according to the Erdős--Rényi $G(n,p)$ model. Due to the in-phase separation, $r_{(mk+l)}^{(i)}$ is independent of other $r_{(mh+l)}^{(i)}$ for $h\neq k$. Additionally having access to network-wide signal is independent of the actions taken and the signals generated.  

Based on \eqref{eq:hattheta} and the randomness on $\mathbb{I}_{\{a_{mk + l} = a\}} \cdot s_{(mk+l)} \cdot r_{(mk+l)}^{(i)}\in \{0,1\}$  independent across $k$, the Hoeffding bound yields that
\be\label{eq:Hoeffding}
\Pr(|\hat{\theta}_{(\ell),l}^{(i)}(a) - \theta(a)| \geq \zeta) \leq 2e^{-2\ell \zeta^2}
\ee
for $l=0,1,\ldots,m-1$, where
\be
\theta(a) := \left\{\begin{array}{ll}
\frac{q_m^i}{|A|} \xi^{C-\eSW(a)} &\mbox{if } a\in \efA_{\epsilon}^{\mathrm{eq}}\\
0 & \mbox{if } a\notin \efA_{\epsilon}^{\mathrm{eq}}
\end{array}\right..
\ee

Observe that
\be
\argmax_{a} \{\theta(a)\} = \argmax_{a\in \efA_{\epsilon}^{\mathrm{eq}}} \{\eSW(a)\}
\ee
since $\xi\in(0,1)$ and $C\geq \eSW(a)$ for all $a$. The following lemma characterizes the separation between $\theta(\efa^*)$ and others $\theta(a)$ for $a\neq \efa^*$.

\begin{lemma}\label{lem:zeta}
Under Assumption \ref{assm:equilibrium}, we have
\be
\theta(\efa^*) - \zeta > \sum_{a\in \efA_{\epsilon}^{\mathrm{eq}}\setminus \efa^*} (\theta(a) + \zeta),
\ee
where
\be
\zeta:= \frac{q_m^i \delta}{|A||\efA_{\epsilon}^{\mathrm{eq}}|} \xi^{C - \max_{a\in \efA_{\epsilon}^{\mathrm{eq}}\setminus \efa^*}\{\eSW(a)\}}.
\ee
\end{lemma}

\begin{proof}
The proof follows from \cite[Lemma 1]{ref:Kiremitci26a} with additional scaling factor $q_m^i$.
\end{proof}

The empirical average of joint action $a$ when there is network-wide content endorsement and agent $i$ can access to that network-wide signal is given by
\be
\hat{\theta}_{(\ell \kappa)}^{(i)}(a) := \frac{1}{m}\sum_{l=0}^{m-1} \hat{\theta}_{(\ell),l}^{(i)}(a).
\ee
Correspondingly, the local empirical average is given by
\be
\hat{\theta}_{(\ell \kappa)}^{i}(a^i) = \sum_{\tilde{a}\in A: \tilde{a}^i = a^i}\hat{\theta}_{(\ell \kappa)}^{(i)}(\tilde{a}).
\ee 

The following proposition shows that agent $i$ can locally identify the best action.

\begin{proposition}\label{prop}
Under Assumptions \ref{assm:equilibrium} and \ref{assm:estimation}, if $|\hat{\theta}_{(k)}^{(i)}(a) - \theta(a)|<\zeta$, where $\zeta$ is as described in Lemma \ref{lem:zeta}, then
\be
\hat{\theta}_{(\ell \kappa)}^{i}(\efa^{i,*}) > \hat{\theta}_{(\ell \kappa)}^{i}(a^i)\quad\forall a^i\neq \efa^{i,*}.
\ee
\end{proposition}

\begin{proof}
The proof follows from \cite[Proposition 1]{ref:Kiremitci26a}.
\end{proof}

Lastly, by \eqref{eq:Hoeffding}, the triangle inequality and the union bound yield that
\be
\Pr(|\hat{\theta}_{(K)}^{(i)}(a) - \theta(a)| \geq \zeta) \leq 2me^{-2K \zeta^2/\kappa}.
\ee
Correspondingly, we have
\be\label{eq:lastbound}
Pr(|\hat{\theta}_{(K)}^{(i)}(a)-\theta(a)| \leq \zeta,\forall a\in A) \geq 1-\beta
\ee
where $\beta := 2m|\efA_{\epsilon}^{\mathrm{eq}}|e^{-2K\zeta^2/\kappa}$. 

By Proposition \ref{prop} and \eqref{eq:lastbound}, we have
\be
\Pr(\hat{a}_{(K)}\neq \efa^*)\leq \beta.
\ee
Then, choosing $K= \frac{1}{2\zeta^2}\log(4m|\efA_{\epsilon}^{\mathrm{eq}}|T\zeta^2)$ yields that $\beta = \frac{1}{2T\zeta^2}$, which completes the proof.
\end{proof}

\section{NUMERICAL RESULTS}
\label{sec:sim_game}

\begin{figure*}[!t]
    \centering

    \begin{subfigure}[t]{0.48\textwidth}
        \centering
        \includegraphics[width=\linewidth]{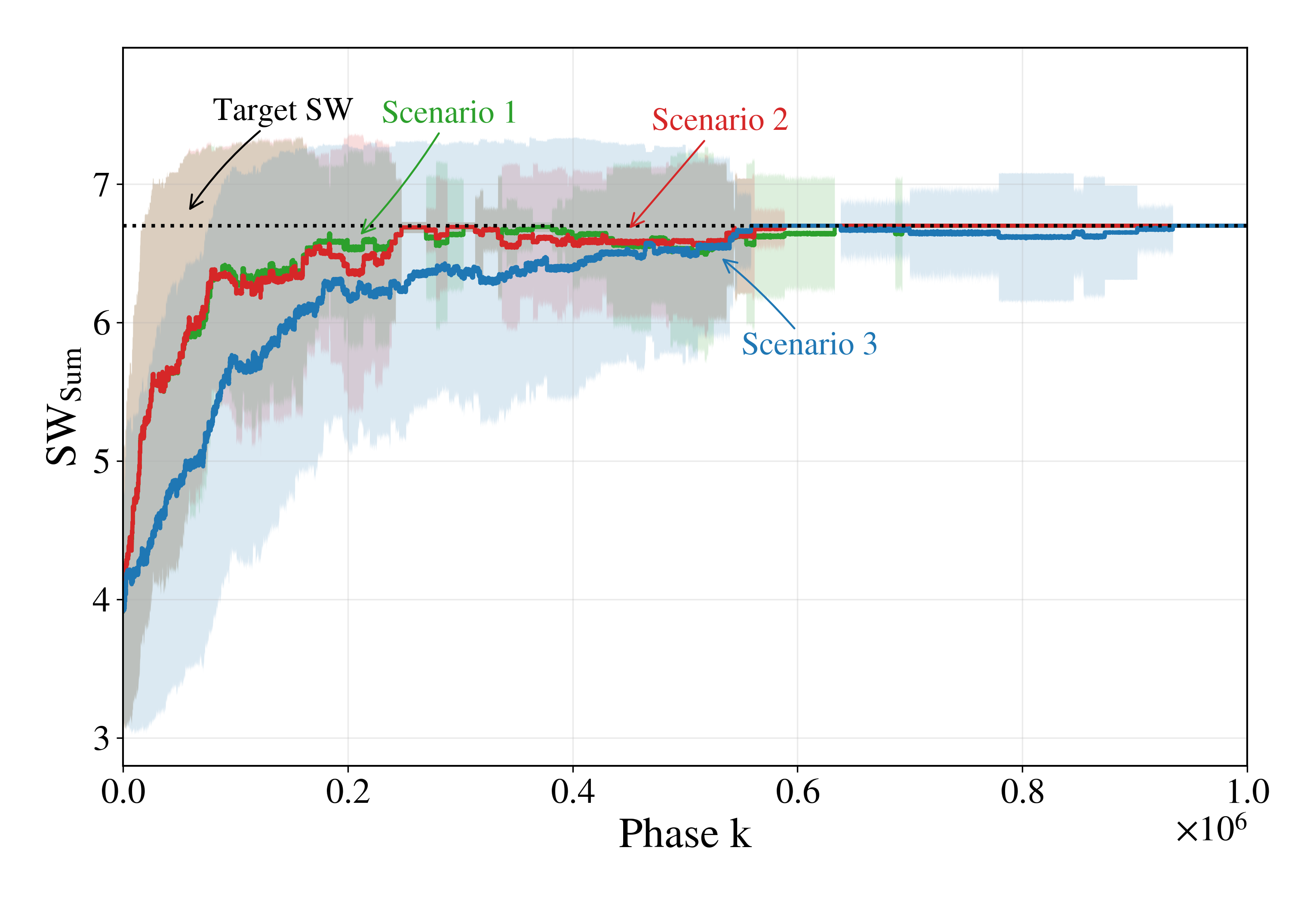}
        \caption{Sum of Local Payoffs as Social Welfare}
        \label{fig:gm_alpha0}
    \end{subfigure}
    \hfill
    \begin{subfigure}[t]{0.48\textwidth}
        \centering
        \includegraphics[width=\linewidth]{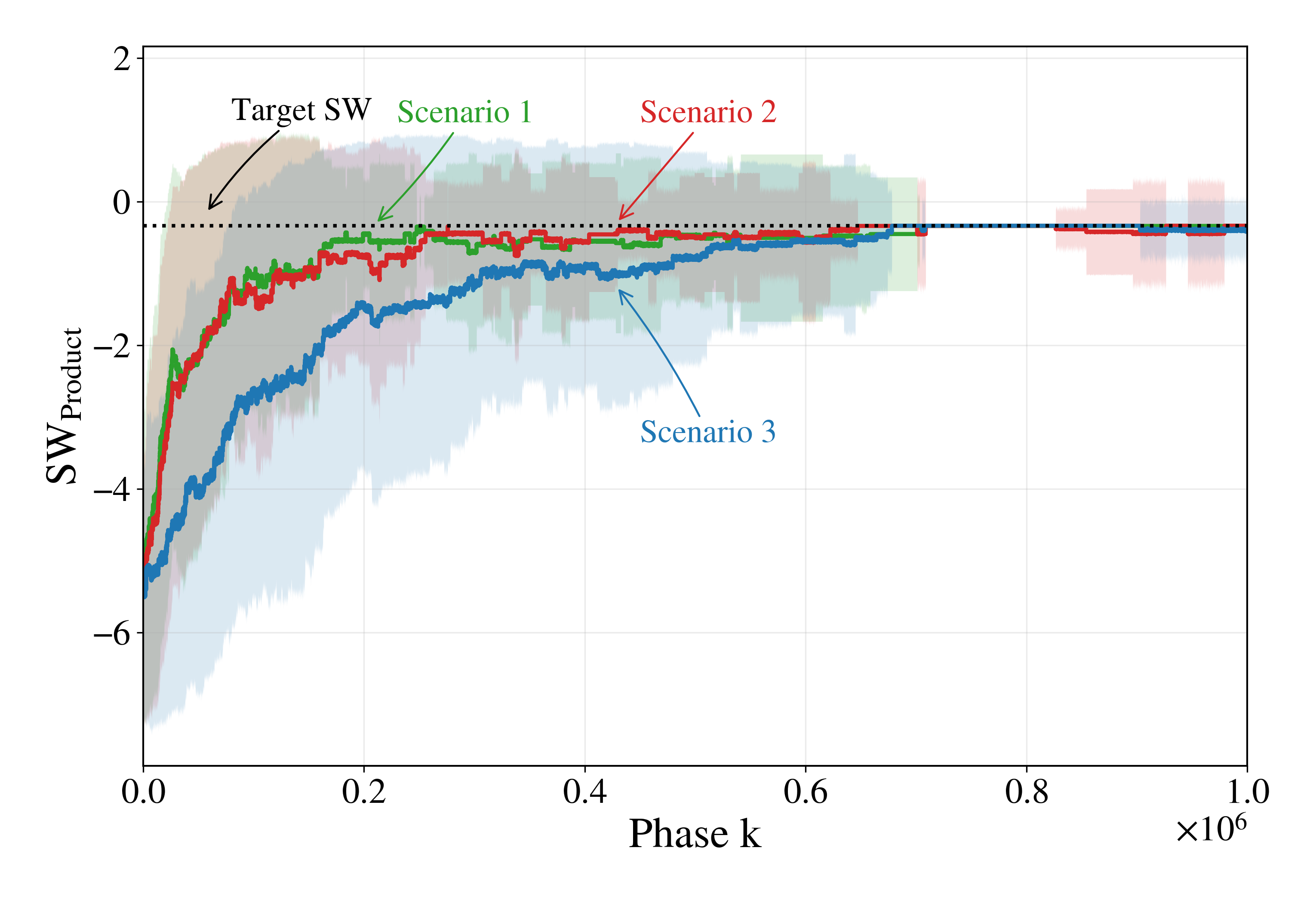}
        \caption{Product of Local Payoffs as Social Welfare}
        \label{fig:gm_alpha1}
    \end{subfigure}
    \medskip


    \begin{subfigure}[t]{0.48\textwidth}
        \centering
        \includegraphics[width=\linewidth]{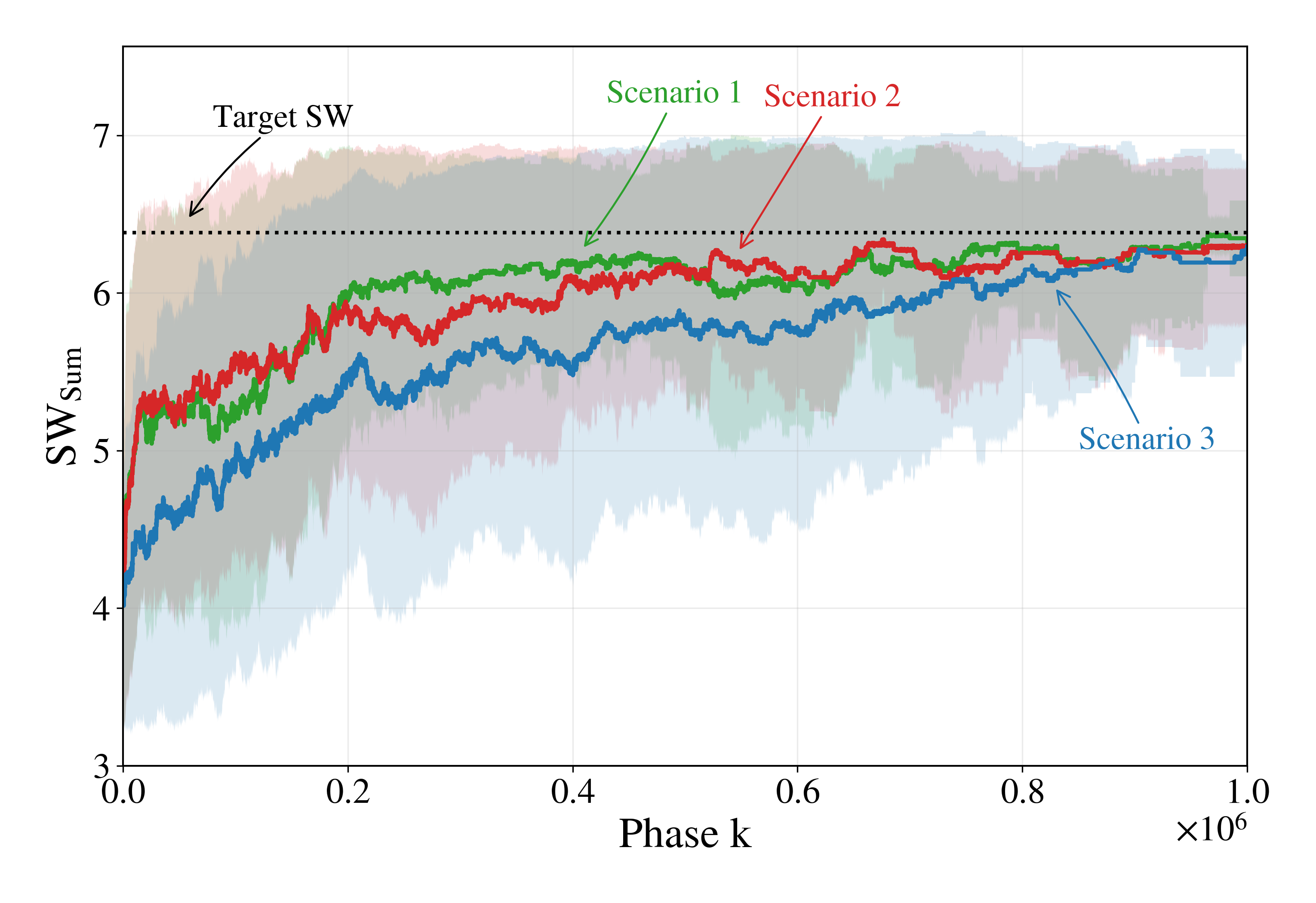}
        \caption{Sum of Local Payoffs as Social Welfare}
        \label{fig:eq_alpha0}
    \end{subfigure}
    \hfill
    \begin{subfigure}[t]{0.48\textwidth}
        \centering
        \includegraphics[width=\linewidth]{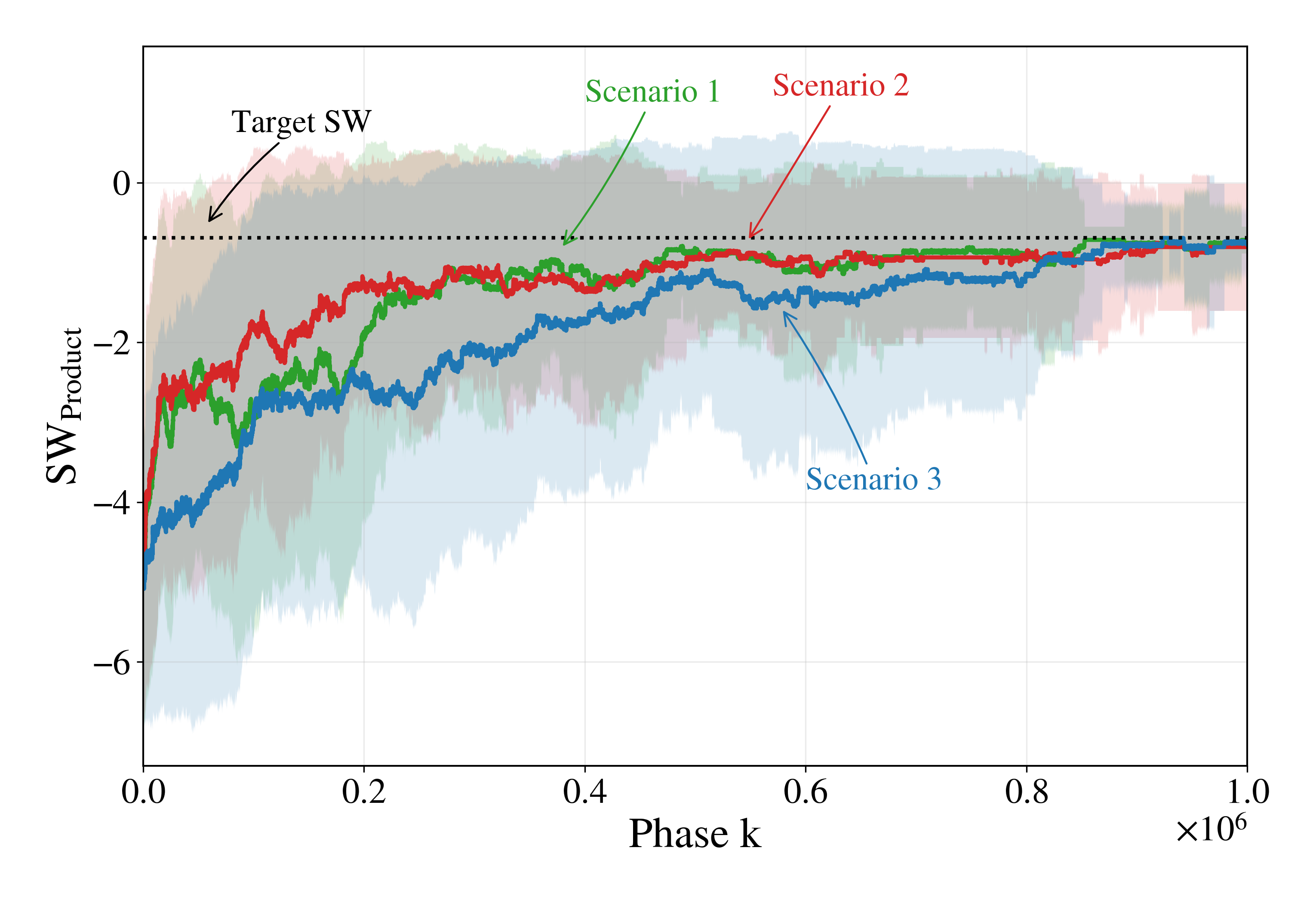}
        \caption{Product of Local Payoffs as Social Welfare}
        \label{fig:eq_alpha1}
    \end{subfigure}

    
 \caption{Comparison of the proposed networked learning scheme for optimal equilibrium selection and social welfare maximization under three communication scenarios. Panels (a) and (c) use utilitarian welfare (sum of local payoffs), while panels (b) and (d) use proportional-fair welfare (sum of log-payoffs). \textcolor{darkgreen}{Scenario~1} is a fully connected baseline (ideal communication), i.e., the same communication setting as the original DOEL dynamics. \textcolor{red}{Scenario~2} and \textcolor{blue}{Scenario~3} use the \textit{restricted Erd\H{o}s--R\'enyi} graph (edge probability $p=0.2$, conditioned on connectivity) with per-phase independent link drops with probability $0.1$ for each edge, and differ in the table depth: $m=5$ (\textcolor{red}{Scenario~2}) versus $m=3$ (\textcolor{blue}{Scenario~3}), so that $m$ is above versus possibly below the backbone diameter $D\in\{3,4\}$. The dashed horizontal line indicates the target welfare level; curves show averages over 50 runs with shaded $\pm 1$ standard deviation.}
    \label{fig:combined_four_panel}
\end{figure*}

We consider the repeated play of a normal-form game with $n=7$ agents and $|A^i|=3$ actions per agent. 
The payoffs are generated randomly over $[0,1]$ once and fixed throughout all experiments. For the realized game instance, the induced equilibrium landscape contains three distinct pure Nash equilibria, each achieving a different value of the social welfare objective. In addition, the (unconstrained) social maximizer,
$a^{\mathrm{SW}}\in \arg\max_{a\in A}\; \SW(a)$,
is unique and does not coincide with any of the equilibrium action profiles, highlighting a nontrivial gap between stability (equilibrium) and efficiency (social welfare).

On this same payoff matrix we evaluate both equilibrium selection and unconstrained welfare maximization, under two welfare specifications: (i) the utilitarian case $\SW_{\mathrm{sum}}(a)=\sum_{i=1}^7 u^i(a)$ and (ii) the proportional-fair case $\SW_{\mathrm{product}}(a)=\sum_{i=1}^7 \log(u^i(a))$ (with a small positive floor applied to avoid $\log(0)$). 
The same networked learning dynamics is then run in the identical environment under both cases, and performance is reported with respect to (a) the optimal equilibrium under the chosen $\SW(\cdot)$ among the three equilibria, and (b) the unique social maximizer $a^{\mathrm{SW}}$ as a separate efficiency benchmark.

We remark that the joint action space has the size of $|A|=2187$, i.e., there are $2187$ possible joint action profiles in the stage game.

We consider a two-stage communication model which we called \textit{restricted Erdős--Rényi}. First, a fixed Erdős--Rényi backbone graph is generated at initialization of each run and retained throughout that run, so the 50 independent runs use different backbones. Pairs that do not share a backbone edge are never connected. The backbone is sampled so that its diameter is $D \in \{3,4\}$. Then, at each phase, each backbone edge is independently dropped with a probability, yielding the active communication graph for that phase. Communication remains neighbor-only; however, links are unreliable over time: at the end of each phase, each edge in the base graph is independently active with a probability.

In all simulations, we use \textit{restricted Erdős--Rényi} with edge probability $p=0.2$ (conditioned on connectivity), and keep this base topology fixed throughout learning with each edge drop probability $0.1$. Since communication occurs only once per phase (at the phase boundary), the active neighbor set is re-sampled once per phase and remains fixed during the $\kappa$ within-phase interaction rounds. Hence, the set of available communication links varies across phases due to random link drops, even though the underlying base graph is time-invariant. We run the process for $K=10^6$ phases.

To study equilibrium selection, we set the phase length to $\kappa=250$ (so the total number of stage-game rounds is $T=\kappa K=2.5\times 10^8$), use $\varepsilon=0.1$ exploration (implemented as $\varepsilon$-greedy), and set the near-best-response tolerance to $\epsilon^i=\epsilon=10^{-5}$ for all agents. The semantic randomization parameter is as $\xi=0.35$ and $\xi=0.4$ for payoff-sum and payoff-product welfares, respectively. 

To obtain the social-welfare maximizer benchmark, we do not divide the time into phases, i.e., $\kappa=1$, $\varepsilon=0$, and $\epsilon^i=+\infty$ (so every action is tolerable), and choose $\xi$ small so that the induced stationary behavior concentrates on the welfare maximizer. Concretely, we use $\xi=10^{-4}$ for payoff-sum maximization and $\xi=1.6\times 10^{-4}$ for payoff-product maximization. As above, we keep the same horizon $K=10^6$.

Fig.~\ref{fig:combined_four_panel} compares three communication scenarios. The plotted curves show the average achieved welfare over 50 independent runs, and the shaded regions indicate $\pm1$ standard deviation. The dashed horizontal line marks the target welfare level in each panel. Scenario~1 is a fully connected baseline, i.e., the same communication setting as the original DOEL dynamics \cite{ref:Kiremitci26b}.

A comparison with the no-communication payoff-based dynamics of \cite{ref:Marden14,ref:Pradelski12} is not included, because the instance is randomly sampled to represent a general normal-form game, and such a draw need not satisfy the structural assumptions those methods require. Scenarios~2 and~3 operate over the restricted Erd\H{o}s--R\'enyi network with link drops, and differ only in the table depth: $m=5$ for Scenario~2 and $m=3$ for Scenario~3.

On a backbone of diameter $D$, a content bit reaches every agent within $m$ phases if $m\ge D$, unless the same agent is dropped for $D$ consecutive phases. Independent drop probability $0.1$ makes that event rare, so Scenario~2 ($m=5\ge D$) recovers network-wide signals with high probability, whereas Scenario~3 ($m=3$, which can be smaller than $D$) can leave some messages unrecovered. This is the experimental counterpart of $q_m$ in Theorem~1. Unlike a fully connected network, some agent pairs are never linked, and independent dropouts can leave an agent disconnected from the rest of the network in a given phase; reconstruction still proceeds through the table window of length $m$.

In the welfare-maximization setting shown in Fig.~\ref{fig:combined_four_panel}(\subref{fig:gm_alpha0}) and Fig.~\ref{fig:combined_four_panel}(\subref{fig:gm_alpha1}), the achieved welfare increases rapidly in the early phases and stabilizes near the target level. Scenario~1 converges fastest and with the smallest variability. Scenario~2 stays close to Scenario~1: even though some agents are never directly linked and may be disconnected from the network in a given phase, their information is still recovered through the table window. Scenario~3 may not recover all information, since the table depth can be smaller than the backbone diameter, and is slower and more variable, but still reaches the target.

A similar trend is observed in the equilibrium-selection setting in Fig.~\ref{fig:combined_four_panel}(\subref{fig:eq_alpha0}) and Fig.~\ref{fig:combined_four_panel}(\subref{fig:eq_alpha1}), where convergence is slower and the separation between scenarios is more pronounced. This is expected since equilibrium selection restricts the feasible set to equilibria, unlike unconstrained welfare maximization. Still, both Scenario~2 and Scenario~3 approach the target equilibrium welfare, with Scenario~2 again closer to Scenario~1.

\section{CONCLUSION}
\label{sec:conclusion}

We studied decentralized learning of socially optimal equilibria in repeated normal-form games under dynamic communication networks and proposed Networked-DOEL, which combines low-bandwidth semantic signaling with time-stamped, time-stacked information exchange. We established logarithmic regret guarantees with an in-phase exploration perturbation and showed that the method remains effective under intermittent connectivity, with temporal redundancy improving robustness. Overall, the results highlight that scalable and resilient equilibrium selection is achievable with minimal communication. The dynamics remain resilient to intermittent links in the sense that a dropped or delayed bit can still be reconstructed inside a table window of length $m$. Promising directions for future research include developing more efficient exploration mechanisms, extending the framework to stochastic games, and incorporating function approximation methods to handle large-scale action spaces.







\bibliographystyle{IEEEtran}
\bibliography{ref}




\end{document}